\documentclass[a4paper,12pt]{article}

\usepackage{amsthm,amsmath,amssymb} 
\usepackage{url} 
\usepackage{cases} 
\usepackage{tikz}
\usepackage{tikz-cd}
\usepackage{float}
\usepackage{xcolor}
\usepackage{algorithm}
\usepackage{algpseudocode}

\theoremstyle{plain} 
\newtheorem{lemma}{Lemma}[section] 
\newtheorem{theorem}[lemma]{Theorem}

\theoremstyle{definition}
\newtheorem{definition}{Definition}[section] 
\newtheorem{remark}{Remark}
\newtheorem{example}{Example}

\title{On the equivalence between additive and linear codes}
\author{Kanat Abdukhalikov \\	
Department of Mathematical Sciences, \\
UAE University, PO Box 15551, Al Ain, UAE\\
Email: abdukhalik@uaeu.ac.ae \bigskip  \\  
Duy Ho \\
Department of Mathematical Sciences, \\
UAE University, PO Box 15551, Al Ain, UAE\\
Email: duyho@uaeu.ac.ae
}
\date{ }

\begin{document} 

\maketitle

\begin{abstract} 
Additive codes have attracted considerable attention for their potential to outperform linear codes. 
However, distinguishing strictly additive codes from those that are equivalent to linear codes remains a fundamental challenge. 
To resolve this ambiguity, we introduce a deterministic test that requires only the generator matrix of the code.
We apply this test to verify the strict additivity of several quaternary additive codes recently reported in the literature. 
Conversely, we demonstrate that a previously known additive complementary dual (ACD) code is equivalent to a linear Hermitian LCD code, thereby improving the best-known bounds for such linear codes.

\end{abstract}



\section{Introduction}

Additive codes have long been a subject of significant research interest, primarily due to their  application in the construction of binary quantum codes, first shown in \cite{calderbank1998}. Additive codes have recently attracted renewed attention because they can provide examples of codes that outperform their linear counterparts, as demonstrated in  \cite{bierbrauer2021, guan2023, kurz2024}.

The relationship between additive and linear codes remains an active area of research. In \cite{bierbrauer2021}, it was established that every optimal quaternary additive code of dimension $k \leq 3$ shares the same parameters as an optimal linear quaternary code. However, recent developments have demonstrated that this equivalence does not hold for higher dimensions. Notably, \cite{ball2026} and \cite{kurz2024} constructed new additive codes with minimum distances exceeding the Griesmer bound for the corresponding linear codes,   proving that they are strictly additive and not equivalent to any linear code.

A significant contribution to this field is the work of Guan, Li, Liu, and  Ma \cite{guan2023},  who constructed numerous quaternary additive codes based on quasi-cyclic codes possessing higher minimum distances than the best-known linear codes of the same length and dimension.
This development raises a fundamental question: when a new additive code is constructed, how can one  determine whether it is  equivalent to an existing linear code with the same parameters? In \cite{ball2026}, Ball and Popatia noted that it remains unclear whether linear codes exist with the same parameters as those presented in \cite{guan2023}.

In this paper,  we provide an answer to this question by formulating a deterministic   test to determine whether an additive code is equivalent to a linear code. 
Our approach is purely algebraic and requires only the generator matrix of the code.
Using this framework, we computationally verify that the additive codes constructed in \cite{guan2023} and \cite{zhu2024}
are not equivalent to any linear codes.  
Furthermore, in  \cite{carlet2018}, it was established that for $q > 4$, every linear code is equivalent to a Hermitian LCD code. 
Motivated by this result, we investigate whether quaternary additive complementary dual (ACD) codes constructed in \cite{guan2023} can be equivalent to linear Hermitian LCD codes. 
We provide one such example where a quaternary ACD code is shown to be equivalent to a linear Hermitian LCD code.

To the best of our knowledge, existing methods for distinguishing strictly additive codes from those that are equivalent to linear codes rely primarily on geometric techniques, as seen in \cite{ball2023, ball2023b,  ball2025, ball2026, blokhuis2004}. 
In the case of MDS codes, Adriaensen and Ball \cite{ball2023} established sufficient conditions under which an additive MDS code must be equivalent to a linear code, using geometric techniques based on pseudo-arcs and Desarguesian spreads.
In contrast, we provide an algebraic method based on the generator matrix that applies to any additive code.

The content of the paper is organized as follows. 
In Section 2, we recall preliminary results and definitions. 
In Section 3, we develop the algebraic framework for the equivalence test, including characterizations of when an additive code is equivalent to a linear code, and present a deterministic algorithm based on the generator matrix of the code.
In Section 4, we demonstrate the applicability of this test to known  additive codes.  
 
All computations in this paper have been done with the computer algebra system MAGMA V2.29-7 \cite{magma}. The MAGMA scripts used for the computations are available from the authors upon request.

\section{Preliminaries}
Let $\mathbb{F}_q$ be a finite field of $q$ elements, and let $\mathbb{F}_{q^2}^* = \langle \omega \rangle$, where $\omega$ is a primitive element of $\mathbb{F}_{q^2}$. We denote by $\mathbb{F}_{q}^{n \times m}$ the $\mathbb{F}_q$-vector space of matrices of size $n \times m$ with entries in $\mathbb{F}_q$. While the algebraic framework introduced in this paper can be naturally generalized to additive codes over $\mathbb{F}_{q^h}$ for any $h \ge 2$, we restrict our attention to the case $h=2$ to directly address the significant interest in quaternary codes found in the recent literature.

\subsection{Additive codes}
We fix $\{1, \omega\}$ as a basis of $\mathbb{F}_{q^2}$ over $\mathbb{F}_q$. Hence, any element $z_i \in \mathbb{F}_{q^2}$ can be uniquely written as $z_i = x_i + \omega y_i$ for $x_i, y_i \in \mathbb{F}_q$. For any vector $z = (z_1, z_2, \dots, z_n) \in \mathbb{F}_{q^2}^n$, we define $\phi : \mathbb{F}_{q^2}^n \to \mathbb{F}_q^{2n}$ as
\begin{equation*}
\phi(z) = (x_1, y_1 \mid x_2, y_2 \mid \dots \mid x_n, y_n),
\end{equation*}
where the vertical bars separate the $n$ coordinate pairs. 

We follow the notation of additive codes from \cite{ball2026} and \cite{kurz2024}. An $[n, k/2]_q^2$ additive code $C$ is an $\mathbb{F}_q$-linear subgroup of $\mathbb{F}_{q^2}^n$ of size $q^k$. Its image $\phi(C)$ is a $k$-dimensional $\mathbb{F}_q$-linear subspace of $\mathbb{F}_q^{2n}$. 
The subspace $\phi(C)$ can be described in terms of a \textit{generator matrix} $G \in \mathbb{F}_q^{k \times 2n}$ of rank $k$, whose rows form an $\mathbb{F}_q$-basis for $\phi(C)$.
We then define the generator matrix of the additive code $C$ to be this same matrix $G$. 

The Hamming distance $d_H(u, v)$ between two vectors $u, v \in \mathbb{F}_{q^2}^n$ is defined as the number of coordinates in which they differ. The Hamming weight $\text{wt}_H(z)$ of a vector $z = (z_1, \dots, z_n) \in \mathbb{F}_{q^2}^n$ is defined as the number of non-zero coordinates, i.e., $\text{wt}_H(z) = |\{ 1 \le i \le n \mid z_i \neq 0 \}|$. Because an additive code is closed under  addition, its minimum Hamming distance is equal to the minimum Hamming weight among all its non-zero codewords. An $[n, k/2]_q^2$ additive code with a minimum Hamming distance $d$ is denoted as an $[n, k/2, d]_q^2$ code.  Thus, an $[n, k, d]_q^1$ code is a standard linear code. Under the map $\phi$, the condition $z_i \neq 0$ holds if and only if the corresponding pair $(x_i, y_i) \neq (0,0)$ in $\mathbb{F}_q^{2n}$.

\subsection{Equivalence of additive codes}

We first define the standard notion of equivalence for additive codes over $\mathbb{F}_{q^2}$, see for example in \cite[Section 9.3]{bierbrauerconcise}. 

\begin{definition}[Monomial Equivalence] \label{monequiv}
Let $C, D \subseteq \mathbb{F}_{q^2}^n$ be two additive codes.
We say that $C$ and $D$ are \textit{monomially equivalent} if there exist a permutation $\pi \in S_n$ and elements $A_i \in \mathrm{GL}_2(\mathbb{F}_q)$ such that
\[
(z_1, z_2, \dots, z_n) \in C \implies (z_{\pi(1)} A_1, z_{\pi(2)} A_2, \dots, z_{\pi(n)} A_n) \in D
\]
for all $z = (z_1, z_2, \dots, z_n) \in C$. Here, we view each $A_i$ as an $\mathbb{F}_q$-linear automorphism of $\mathbb{F}_{q^2}$ acting on the right, where the elements of $\mathbb{F}_{q^2}$ are represented as row vectors with respect to the basis $\{1, \omega\}$. 

In terms of generator matrices, $C$ and $D$ are monomially equivalent if there exist $\mathbb{F}_q$-generator matrices $G$ and $G'$ of their respective images $\phi(C)$ and $\phi(D)$, an $n \times n$ permutation matrix $P$, and a block-diagonal matrix $A = \mathrm{diag}(A_1, A_2, \dots, A_n)$ with $A_i \in \mathrm{GL}_2(\mathbb{F}_q)$ for $1 \le i \le n$, such that
\begin{equation}
G' = G A (P \otimes I_2).
\end{equation}
\end{definition}

Here  and throughout the paper, codewords are row vectors and all matrices act on the right. The rows of $G$ span $\phi(C)$, and right multiplication by $A(P \otimes I_2)$ acts on the $2n$ coordinates.

When we say two codes are equivalent, we mean that they are monomially equivalent. We are interested in  determining whether a given additive code is equivalent to a linear code over $\mathbb{F}_{q^2}$. We say that an additive code is \textit{strictly additive} if it is not equivalent to any linear code over $\mathbb{F}_{q^2}$.

\begin{definition}[$\mathrm{SL}$-Equivalence] \label{slequiv}
Two additive codes $C, D \subseteq \mathbb{F}_{q^2}^n$ are \textit{$\mathrm{SL}$-equivalent} if they are monomially equivalent and the matrices $A_i$ in Definition \ref{monequiv} belong to $\mathrm{SL}_2(\mathbb{F}_q)$ instead of $\mathrm{GL}_2(\mathbb{F}_q)$.
\end{definition}

\begin{remark} \label{remarkeven}
By \cite[Lemma 13.3.8]{mullen2013}, any matrix $A_i \in \mathrm{GL}_2(\mathbb{F}_q)$ can   be factored as $A_i = B_i D_i$, where $B_i \in \mathrm{SL}_2(\mathbb{F}_q)$ and $D_i = \mathrm{diag}(1, \det(A_i))$ is a diagonal matrix. Applying the $\mathbb{F}_q$-linear automorphism corresponding to right-multiplication by $D_i$ scales only one component of the row vector representation, and so it does not preserve $\mathbb{F}_{q^2}$-linearity. Consequently, for odd $q$, one must consider the full monomial equivalence.

However, when $q$ is even, we can obtain a more convenient factorization. Because every element in a field of even characteristic is a perfect square, there always exists a unique $c_i \in \mathbb{F}_q^*$ such that $c_i^2 = \det(A_i)$. 
This allows us to factor $A_i$ instead as $A_i = B_i D_i$, where $B_i \in \mathrm{SL}_2(\mathbb{F}_q)$ and $D_i = c_i I_2$ is a scalar matrix. In this case, right-multiplication by $D_i$ acts on the $i$-th coordinate as scalar multiplication by $c_i \in \mathbb{F}_q^* \subset \mathbb{F}_{q^2}^*$. Applying such coordinatewise scalar multiplications to an $\mathbb{F}_{q^2}$-linear code again gives an $\mathbb{F}_{q^2}$-linear code.  Thus, when $q$ is even, it is sufficient to check only for $\mathrm{SL}$-equivalence.
Conversely, applying a monomial equivalence transformation to an $\mathbb{F}_{q^2}$-linear code can result in an additive code that is not $\mathbb{F}_{q^2}$-linear (for example, see \cite[Exercise 558]{huffman2010}).
\end{remark}

\subsection{Inner products and ACD codes}

For the applications to ACD codes in Section 4, we recall some material on inner products and hulls of additive codes. 
For $\mathbf{x} = (x_1, \dots, x_n)$ and $\mathbf{y} = (y_1, \dots, y_n)$ in $\mathbb{F}_{q^2}^n$, the \textit{alternating form} is defined by
\begin{equation*}
\langle \mathbf{x}, \mathbf{y} \rangle_a = \sum_{i=1}^n \frac{x_i y_i^q - x_i^q y_i}{\omega^q - \omega}.
\end{equation*}
For $\mathbf{u} = (a_1, b_1 \mid \dots \mid a_n, b_n)$ and $\mathbf{v} = (c_1, d_1 \mid \dots \mid c_n, d_n)$ in $\mathbb{F}_q^{2n}$, the \textit{symplectic inner product} is defined as
\begin{equation*}
\langle \mathbf{u}, \mathbf{v} \rangle_s = \sum_{i=1}^n (a_i d_i - b_i c_i).
\end{equation*}
It can be checked that $\langle \mathbf{x}, \mathbf{y} \rangle_a = \langle \phi(\mathbf{x}), \phi(\mathbf{y}) \rangle_s$ for any $\mathbf{x}, \mathbf{y} \in \mathbb{F}_{q^2}^n$. 
The \textit{dual} of an additive code $C \subseteq \mathbb{F}_{q^2}^n$ with respect to the alternating form is defined as
\begin{equation*}
C^{\perp_a} = \{ \mathbf{y} \in \mathbb{F}_{q^2}^n \mid \langle \mathbf{x}, \mathbf{y} \rangle_a = 0 \text{ for all } \mathbf{x} \in C \}.
\end{equation*}
The \textit{hull} of $C$ is defined as $\mathrm{Hull}_a(C) = C \cap C^{\perp_a}$. An additive code $C$ is called an \textit{additive complementary dual} (ACD) code if its hull is trivial, i.e., $\mathrm{Hull}_a(C) = \{\mathbf{0}\}$. Because $\phi$ preserves the inner product, an additive code $C \subseteq \mathbb{F}_{q^2}^n$ is an ACD code if and only if its image $\phi(C)$ is a symplectic linear complementary dual (LCD) code (that is, $\phi(C) \cap \phi(C)^{\perp_s} = \{\mathbf{0}\}$).

\begin{remark} \label{remarkSLhull}
Our consideration of \textrm{SL}-equivalence is motivated by the characterization of symplectic isometries over local Frobenius rings by Gluesing-Luerssen and Pllaha \cite[Theorem 7.1]{gluesing2019}. As noted in \cite[Section 7]{gluesing2019}, symplectic isometries preserve both the symplectic inner product and the symplectic distance. Since finite fields are local Frobenius rings, two additive codes are \textrm{SL}-equivalent if and only if they are related by a symplectic isometry.
Furthermore, since \textrm{SL}-equivalence   preserves the symplectic inner product, any two \textrm{SL}-equivalent codes have the same hull dimension.  
\end{remark}
 
\section{Equivalence test for linearity of codes}

Throughout this section, let $C \subseteq \mathbb{F}_{q^2}^n$ be an $[n, k/2]_q^2$ additive code. 
We recall that $k$ denotes the $\mathbb{F}_q$-dimension of $C$, so that $C$ has size $q^k$ and its generator matrix is $G \in \mathbb{F}_q^{k \times 2n}$.
The parameter $k/2$ in the notation $[n, k/2]_q^2$ need not be an integer. 
An $\mathbb{F}_{q^2}$-linear code of $\mathbb{F}_{q^2}$-dimension $m$ has $\mathbb{F}_q$-dimension $2m$. Since equivalence preserves the size of a code, an additive code with odd $k$ is never equivalent to a linear code. 
We therefore assume from now on that $k$ is even.

Let the minimal polynomial of $\omega$ over $\mathbb{F}_q$ be $f(x)=x^2+c_1x+c_0$. 
Following \cite[Section 13.2]{mullen2013}, we define the $2 \times 2$ companion matrix of $\omega$ over $\mathbb{F}_q$ as
\[ \mathbf{M}_\omega = \begin{pmatrix} 0 & 1 \\ -c_0 & -c_1 \end{pmatrix}. \]
Since $\omega$ generates the multiplicative group $\mathbb{F}_{q^2}^*$, the matrix $\mathbf{M}_\omega$ generates a cyclic subgroup of $\textrm{GL}_2(\mathbb{F}_q)$ of order $q^2-1$, commonly known as a \textit{Singer group} (or Singer cycle). The mapping $\Phi: \mathbb{F}_{q^2} \to \mathbb{F}_q^{2 \times 2}$ uniquely determined by $\Phi(\omega) = \mathbf{M}_\omega$ defines a ring homomorphism, known as the regular representation of the field. 

Under this representation, multiplying by any non-zero element $\alpha \in \mathbb{F}_{q^2}^*$ corresponds to matrix multiplication within the Singer group. Since every non-zero element can be written as $\alpha = \omega^i$ for some integer $i$, its corresponding matrix representation is $\mathbf{M}_\alpha = \mathbf{M}_\omega^i$. Alternatively, writing $\alpha$ in the basis $\{1, \omega\}$ as $\alpha = a + b\omega$, its matrix representation takes the linear form $\mathbf{M}_\alpha = aI_2 + b\mathbf{M}_\omega$.

To extend this action to $\mathbb{F}_q^{2n}$, we define the $2n \times 2n$ block-diagonal matrix $M_\omega = I_n \otimes \mathbf{M}_\omega$.

\begin{theorem} \label{additivecharacterization}
Let $C \subseteq \mathbb{F}_{q^2}^n$ be an additive code of $\mathbb{F}_q$-dimension $k$, and let $G \in \mathbb{F}_q^{k \times 2n}$ be a generator matrix of $C$. 
Then $C$ is equivalent to an $\mathbb{F}_{q^2}$-linear code if and only if there exist $J \in \textrm{GL}_k(\mathbb{F}_q)$ and a block-diagonal matrix $A = \operatorname{diag}(A_1, \dots, A_n)$ with 
$A_i \in \textrm{GL}_2(\mathbb{F}_q)$
such that
\[ J G = G (A M_\omega A^{-1}). \]
\end{theorem}
\begin{proof} 
 Assume that $C$ is equivalent to an $\mathbb{F}_{q^2}$-linear code $C'$.  Since $G$ is a generator matrix of $C$, there exists a generator matrix $G'$ of $C'$ such that $G' = G A(P \otimes I_2)$. Let $M = A(P \otimes I_2)$, and so $G' = GM$. 
Since  $C'$ is $\mathbb{F}_{q^2}$-linear, its image $\phi(C')$ is closed under scalar multiplication by $\omega$. 
 This implies there exists $J \in \mathbb{F}_q^{k \times k}$ such that $J G' = G' M_\omega$. The corresponding linear map is automatically invertible, since multiplication by $\omega$ is invertible, so $J \in \mathrm{GL}_k(\mathbb{F}_q)$.
Substituting $G' = GM$ gives $J (GM) = (GM) M_\omega$, which implies $J G = G(M M_\omega M^{-1})$. Expanding $M$, we have
\begin{align*}
M M_\omega M^{-1} &= A(P \otimes I_2) (I_n \otimes \mathbf{M}_\omega) (P^{-1} \otimes I_2) A^{-1} \\
  &= A(P I_n P^{-1} \otimes I_2 \mathbf{M}_\omega I_2) A^{-1} \\
  &= A(I_n \otimes \mathbf{M}_\omega) A^{-1} \\
  &= A M_\omega A^{-1},
\end{align*} 
and so $J G = G (A M_\omega A^{-1}).$

Conversely, assume that there exist $J \in \textrm{GL}_k(\mathbb{F}_q)$ and $A = \operatorname{diag}(A_1, \dots, A_n)$ 
with $A_i \in \textrm{GL}_2(\mathbb{F}_q)$
such that $J G = G (A M_\omega A^{-1})$.
Let $C'$ be the code generated by $G' = GA$. 
By Definition \ref{monequiv}, the code $C'$ is equivalent  to $C$. 
Then
\[ J G' = J(GA) = (JG)A = G(AM_\omega A^{-1})A = GAM_\omega = G'M_\omega. \]
This shows that $C'$ is an $\mathbb{F}_{q^2}$-linear code.
\end{proof}

For each $i = 1, \dots, n$, let $G_i \in \mathbb{F}_q^{k \times 2}$ denote the $i$-th block column of $G$, formed by the $(2i{-}1)$-th and $2i$-th columns, and let $U_i$ denote the column space of $G_i$.
From the equation $JG = G(AM_\omega A^{-1})$, we have
\[
J G_i = G_i T_i,
\]
where $T_i = A_i \mathbf{M}_\omega A_i^{-1}$. 
If $G_i$ is the zero matrix for some $i$, then the $i$-th coordinate of every codeword is zero, and we may puncture this coordinate. 
A monomial equivalence maps an identically zero coordinate to an identically zero coordinate, and a code with an identically zero coordinate is linear over $\mathbb{F}_{q^2}$ if and only if the code obtained by deleting that coordinate is linear. Hence $C$ is equivalent to a linear code if and only if the punctured code is.

From now on we assume that $G_i$ is nonzero for all $i$.
Define
$$
\mathcal{R} = \{ R \in \mathbb{F}_{q}^{k \times k} : R\, U_i \subseteq U_i \text{ for all } i = 1, \dots, n \}.
$$
It can be checked that $\mathcal{R}$ is an $\mathbb{F}_q$-subspace of $\mathbb{F}_{q}^{k \times k}$. 
 
\begin{theorem}\label{eigenvalue} 
If $C$ is equivalent to a linear code over $\mathbb{F}_{q^2}$, then there exists $J \in \mathcal{R}$ such that $f(J) = J^2 + c_1 J + c_0 I_k = \mathbf{0}$. 
For any such $J$, the algebra 
\[
\mathbb{F}_q[J] = \{aI_k + bJ : a, b \in \mathbb{F}_q\}
\]
is isomorphic to $\mathbb{F}_{q^2}$. 
In particular, $J$ has no eigenvalue in $\mathbb{F}_q$.
Moreover, if $\operatorname{rank}(G_i) = 1$ for some $i$, then $C$ is not equivalent to any linear code over $\mathbb{F}_{q^2}$.
\end{theorem}
 
\begin{proof} 
Assume that $C$ is equivalent to a linear code over $\mathbb{F}_{q^2}$.
By Theorem~\ref{additivecharacterization}, there exists $J \in \mathrm{GL}_k(\mathbb{F}_q)$ such that $JG_i = G_iT_i$ for all $i$, where $T_i = A_i \mathbf{M}_\omega A_i^{-1}$. 
In particular, $J U_i = U_i$ for all $i$, so $J \in \mathcal{R}$. 
Since each $T_i$ is conjugate to $\mathbf{M}_\omega$, we have $T_i^2 + c_1 T_i + c_0 I_2 =\mathbf{0}$ for all $i$.
Then
\[ 
(J^2 + c_1 J + c_0 I_k)\, G_i = G_i\, (T_i^2 + c_1 T_i + c_0 I_2) =\mathbf{0}. 
\]
Since $G$ has rank $k$, the columns of $G_1, \dots, G_n$ span $\mathbb{F}_q^k$, so $f(J)=J^2 + c_1 J + c_0 I_k = \mathbf{0}$.

Now let $J \in \mathcal{R}$ be such that $f(J) = \mathbf{0}$.
The minimal polynomial of $J$ divides $f$, and since $f$ is irreducible over $\mathbb{F}_q$, it is $f$. 
In particular, $J$ is not a scalar matrix, and $I_k$ and $J$ are linearly independent, so the map $a + b\omega \mapsto aI_k + bJ$ is an isomorphism from $\mathbb{F}_{q^2}$ onto $\mathbb{F}_q[J] = \{aI_k + bJ : a, b \in \mathbb{F}_q\}$.  

If $Jv = \lambda v$ for some nonzero $v$ and $\lambda \in \mathbb{F}_q$, then $\mathbf{0} = f(J)v = f(\lambda)v$, implying $f(\lambda)=0$. This contradicts the irreducibility of $f$, and so $J$ has no eigenvalue in $\mathbb{F}_q$.

For the last statement, if $\operatorname{rank}(G_i) = 1$ for some $i$, then every $R \in \mathcal{R}$ preserves the one-dimensional subspace $U_i$, so every $R \in \mathcal{R}$ has an eigenvalue in $\mathbb{F}_q$.
\end{proof}

We now assume that $\operatorname{rank}(G_i) = 2$ for all $i$. We first show that under this assumption,  the first part of Theorem \ref{eigenvalue} can be strengthened to an if and only if statement.

\begin{theorem} \label{minimalpolyR}
Assume that $\operatorname{rank}(G_i)=2$ for all $i$.
Then $C$ is equivalent to an $\mathbb F_{q^2}$-linear code if and only if
there exists $J\in\mathcal R$ such that
\[
f(J)=J^2+c_1J+c_0I_k=\mathbf 0.
\]
\end{theorem}

\begin{proof} 
If $C$ is equivalent to an $\mathbb{F}_{q^2}$-linear code, then by Theorem~\ref{eigenvalue} there exists $J \in \mathcal{R}$ with $f(J) = \mathbf{0}$.
Conversely, assume that there exists $J \in \mathcal{R}$ such that $f(J) = \mathbf{0}$.
For each $i$, since $J U_i \subseteq U_i$ and $\operatorname{rank}(G_i) = 2$,   there exists a unique $T_i \in \mathbb{F}_q^{2 \times 2}$ such that 
$JG_i = G_iT_i.$
Then
\[
\mathbf{0} = f(J)G_i = G_i f(T_i).
\]
Since $G_i$ has rank $2$, we have $f(T_i) = \mathbf{0}$, and so the minimal polynomial of $T_i$ divides $f$. 
As $f$ is irreducible over $\mathbb{F}_q$, the minimal polynomial of $T_i$ is exactly $f$, and so $T_i$ has no eigenvalue in $\mathbb{F}_q$. 
Hence, for any nonzero row vector $v \in \mathbb{F}_q^{1 \times 2}$, the vectors $v$ and $vT_i$ are linearly independent, so the matrix $B_i = \begin{pmatrix} v \\ vT_i \end{pmatrix}$ is invertible. From $T_i^2 = -c_0 I_2 - c_1 T_i$ we obtain
\[
B_i T_i = \begin{pmatrix} vT_i \\ vT_i^2 \end{pmatrix} 
= \begin{pmatrix} vT_i \\ -c_0 v - c_1 vT_i \end{pmatrix} 
= \mathbf{M}_\omega B_i,
\]
and therefore $T_i$ is conjugate to   $\mathbf{M}_\omega$.

Moreover, $J$ is invertible, since $J(J + c_1 I_k) = -c_0 I_k$ and $c_0 \neq 0$. 
Writing $T_i = A_i \mathbf{M}_\omega A_i^{-1}$ with $A_i \in \mathrm{GL}_2(\mathbb{F}_q)$ and $A = \operatorname{diag}(A_1, \dots, A_n)$, we obtain $JG = G(A M_\omega A^{-1})$.
By Theorem~\ref{additivecharacterization}, $C$ is equivalent to an $\mathbb F_{q^2}$-linear code.
\end{proof}

To compute $\mathcal{R}$, we transform the matrix equation $RG_i = G_iT_i$ into a standard homogeneous 
linear system using Kronecker products and vectorization, see for example~\cite[Subsection~2.4]{magnus2019}. 
For any matrix $X$, let $\mathrm{vec}(X)$ denote its vectorization, defined as the single column vector containing the entries of $X$ in column-major order. 
Applying the vectorization identity $\mathrm{vec}(XYZ) = (Z^T \otimes X)\,\mathrm{vec}(Y)$ (see~\cite[Theorem~2.2]{magnus2019}) to both sides of $RG_i = G_iT_i$, we obtain
\[
(G_i^T \otimes I_k)\,\mathrm{vec}(R) - (I_2 \otimes G_i)\,\mathrm{vec}(T_i) = \mathbf{0}.
\]
In other words, we obtain the system $S\,\mathbf{x} = \mathbf{0}$, where
\[
\mathbf{x}= \begin{pmatrix} \mathrm{vec}(R) \\ \mathrm{vec}(T_1) \\ \vdots \\ \mathrm{vec}(T_n) \end{pmatrix},
\]
and $S$ is the $2nk \times (k^2 + 4n)$ matrix
\begin{equation} \label{tildeS}
S = \begin{bmatrix}
G_1^T \otimes I_k & -(I_2 \otimes G_1) & 0 & \cdots & 0 \\
G_2^T \otimes I_k & 0 & -(I_2 \otimes G_2) & \cdots & 0 \\
\vdots & \vdots & \vdots & \ddots & \vdots \\
G_n^T \otimes I_k & 0 & 0 & \cdots & -(I_2 \otimes G_n)
\end{bmatrix}.
\end{equation}
Since $\operatorname{rank}(G_i) = 2$, the matrix $T_i$ is uniquely determined by $R$ from the relation $RG_i = G_iT_i$. Conversely, given any $R \in \mathcal{R}$, defining $T_i$ by this relation gives a vector $\mathbf{x} = (\mathrm{vec}(R), \mathrm{vec}(T_1), \dots, \mathrm{vec}(T_n))^T$ in the null space of $S$. This is a bijection between $\mathcal{R}$ and the null space of $S$. Since both are $\mathbb{F}_q$-vector spaces, we have
$$
\dim_{\mathbb{F}_q}(\mathcal{R}) = \mathrm{nullity}(S).
$$
We observe that $\mathcal{R}$ is a unital $\mathbb{F}_q$-subalgebra of $\mathbb{F}_q^{k \times k}$. We have $I_k \in \mathcal{R}$, and if $R, R' \in \mathcal{R}$, then $RR' U_i \subseteq R\, U_i \subseteq U_i$ for all $i$, so $RR' \in \mathcal{R}$.

\begin{theorem} \label{notSLetest} \label{oddnullity}
Let $C \subseteq \mathbb{F}_{q^2}^n$ be an additive code of $\mathbb{F}_q$-dimension $k$, and let $G \in \mathbb{F}_q^{k \times 2n}$ be a generator matrix of $C$. 
Assume that $\operatorname{rank}(G_i) = 2$ for all $i$, and let $S$ be the matrix defined in~\eqref{tildeS}. 
If $\mathrm{nullity}(S)$ is odd, then $C$ is not equivalent to any linear code over $\mathbb{F}_{q^2}$.

\end{theorem}

\begin{proof}
Suppose for contradiction that $C$ is equivalent to a linear code over $\mathbb{F}_{q^2}$. 
By Theorem~\ref{eigenvalue}, there exists $J \in \mathcal{R}$ with $f(J) = \mathbf{0}$ and $\mathbb{F}_q[J] \cong \mathbb{F}_{q^2}$. 
Since $\mathcal{R}$ is a unital subalgebra containing $J$, it is closed under left multiplication by $\mathbb{F}_q[J]$, and hence $\mathcal{R}$ is a vector space over the field $\mathbb{F}_q[J] \cong \mathbb{F}_{q^2}$. 
Therefore
\[
\mathrm{nullity}(S) = \dim_{\mathbb{F}_q}(\mathcal{R}) = 2 \dim_{\mathbb{F}_{q^2}}(\mathcal{R})
\]
is even, a contradiction.
\end{proof}

 From Theorems~\ref{eigenvalue}, ~\ref{minimalpolyR}, and~\ref{notSLetest}, together with the computation of $\mathcal{R}$ via the matrix $S$, we propose the following Algorithm \ref{equivtest} for determining whether an additive code is equivalent to a linear code. 
 Note that after Step 1, every remaining block $G_i$ has rank $2$. 
 We have this assumption for  Steps 2 and 3, since it guarantees that each $T_i$ is uniquely determined by $R$ and that $\mathrm{nullity}(S)=\dim_{\mathbb{F}_q}(\mathcal{R})$, as required by Theorems~\ref{minimalpolyR} and~\ref{notSLetest}.

\begin{algorithm}[H]
\caption{Equivalence Test for Additive Codes}
\label{equivtest}
\begin{algorithmic}[1]
\Statex \textbf{Input:} A generator matrix $G \in \mathbb{F}_q^{k \times 2n}$ of an additive code $C \subseteq \mathbb{F}_{q^2}^n$,  with respect to the basis $\{1, \omega\}$ of $\mathbb{F}_{q^2}$ over $\mathbb{F}_q$.  
\Statex \textbf{Output:} A determination of whether $C$ is equivalent to a linear code over $\mathbb{F}_{q^2}$.
\Statex
\Statex \textbf{Step 1: Rank check.}
\State Partition $G$ into $n$ block columns $G_1, \dots, G_n$, where each $G_i \in \mathbb{F}_q^{k \times 2}$.
\For{$i = 1, \dots, n$}
    \If{$G_i$ is the zero matrix}
        \State Remove block $G_i$ (puncture the coordinate) and update $n$.
    \EndIf
\EndFor
\For{each remaining block $G_i$}
    \If{$\operatorname{rank}(G_i) = 1$}
        \State \Return ``Not equivalent to any linear code.'' \Comment{Theorem~\ref{eigenvalue}}
    \EndIf
\EndFor
\Statex
\Statex \textbf{Step 2: Nullity test.}
\State Construct the matrix $S$ as defined in~\eqref{tildeS} and compute $d = \mathrm{nullity}(S)$.
\If{$d$ is odd}
    \State \Return ``Not equivalent to any linear code.'' \Comment{Theorem~\ref{notSLetest}}
\EndIf

\Statex \textbf{Step 3: Null space search.}
\State Compute a basis $\{R_1, \dots, R_d\}$ of $\mathcal{R}$ from the null space of $S$.
\State Determine whether there exists $J = \sum_{j=1}^{d} a_j R_j$, with $a_j \in \mathbb{F}_q$, such that $f(J) = \mathbf{0}$, where $f(x) = x^2 + c_1 x + c_0$ is the minimal polynomial of $\omega$ over $\mathbb{F}_q$.
\If{such $J$ exists}
    \State \Return ``Equivalent to a linear code.'' \Comment{Theorem~\ref{minimalpolyR}}
\Else
    \State \Return ``Not equivalent to any linear code.'' \Comment{Theorem~\ref{minimalpolyR}}
\EndIf
\color{black}
\end{algorithmic}
\end{algorithm}

Steps 1 and 2 of Algorithm~\ref{equivtest} require only rank and nullity computations over $\mathbb{F}_q$. 
In Step 3, writing $J = \sum_{j=1}^{d} a_j R_j$ with $d = \mathrm{nullity}(S)$, 
the condition $f(J) = \mathbf{0}$ is a system of quadratic equations in the unknowns 
$a_1, \dots, a_d$, which can be decided by an exhaustive search over at most 
$q^d$ candidates, each requiring a single evaluation of $f(J)$. 
In practice $d$ is small; for all codes in Section~4 that reach Step 3, we have $d = 2$.

\section{Applications of the equivalence test}

In this section, we apply the equivalence test described in Algorithm \ref{equivtest} to computationally determine whether several known quaternary additive codes are  equivalent to linear codes. For our calculations, we let $q = 2$ and let $\omega$ be a primitive element of $\mathbb{F}_4$ satisfying the minimal polynomial $\omega^2 + \omega + 1 = 0$ over $\mathbb{F}_2$. As noted in Remark \ref{remarkeven}, since $q=2$ is   even, it is sufficient to restrict our consideration for the quaternary codes to $\textrm{SL}$-equivalence. 

We begin by considering Example 3 in \cite{guan2023}. While this code was verified using geometric methods in \cite{ball2026}, we provide an independent algebraic verification here.

\begin{example} Let  $n = 63$, and define the ring $R_{2,63} = \mathbb{F}_2[x]/\langle x^{63} - 1 \rangle$.
For $h(x) = h_0 + h_1 x + h_2 x^2 + \cdots + h_{62} x^{62} \in R_{2,63}$, let  
$
[h(x)] = [h_0, h_1, h_2, \dots, h_{62}].
$ 
Define the following polynomials:
\begin{align*}
g(x) &= x^{53} + x^{52} + x^{51} + x^{50} + x^{48} + x^{47} + x^{45} + x^{43} + x^{42} + x^{40} + x^{39} + x^{38} \\
&\quad + x^{31} + x^{28} + x^{25} + x^{24} + x^{21} + x^{20} + x^{19} + x^{17} + x^{14} + x^{13} + x^9 + x^8 \\
&\quad + x^5 + x + 1, \\
f_0(x) &= x^{61} + x^{59} + x^{58} + x^{54} + x^{52} + x^{50} + x^{45} + x^{44} + x^{43} + x^{41} + x^{39} + x^{33} \\
&\quad + x^{32} + x^{31} + x^{26} + x^{24} + x^{23} + x^{22} + x^{20} + x^{18} + x^{13} + x^{12} + x^{11} + x^{10} \\
&\quad + x^9 + x^6 + x^4 + x^2 + x, \\
f_1(x) &= 1.
\end{align*}
Let $C_{QC}$ be the 1-generator quasi-cyclic code over $\mathbb{F}_2$ generated by $([g(x)f_0(x)], [g(x)f_1(x)])$. Let $C$ be the corresponding additive code with parameters $[63, 5,45]^2_2$. 
Using MAGMA \cite{magma}, we compute the generator matrix $G$ of $C$. 
The corresponding matrix $S$ has nullity $1$.
By Theorem \ref{notSLetest}, $C$ is not equivalent to any linear code over $\mathbb{F}_4$.
\end{example}

In addition, we  verified the quaternary additive codes with integral dimensions listed in \cite[Table II]{guan2023}. 
Following the notation in \cite{guan2023}, ``Ex" denotes an extended code, and ``Au" denotes a code augmented by the all-ones vector $\mathbf{1}_n$. The explicit generator polynomials for the $[127, 11, 79]^2_2$ code and its extended code are missing from the original text and  were provided by the authors in a private communication \cite{guanPC}. All the codes that we verified have odd $\operatorname{nullity}(S)$, and so they are not equivalent  to any linear code. 
A summary of this verification is provided in Table \ref{table2_verification}. 
\begin{table}[H]
\centering
\caption{Verification of Quaternary Additive Codes from \cite[Table II]{guan2023}}
\renewcommand{\arraystretch}{1.4}
\begin{tabular}{|c|p{4.5cm}|p{3.5cm}|p{2.5cm}|}
\hline
\textbf{Line} & \textbf{Code Parameters} & \textbf{Description in \cite{guan2023}} & $\text{nullity}(S)$ \\ \hline

1 & $[47,35,7]^2_2$ & Example 5 & 1 \\ \hline

2 & $[56,11,30]^2_2$ (ExAu) & Example 6 & 1 \\ \hline

3 & $[63,5,45]^2_2$  \newline $[64,5,46]^2_2$ (Ex) & Example 3 & 1 \newline 3 \\ \hline

4 & $[91,7,63]^2_2$ (Au) \newline $[92,7,64]^2_2$ (ExAu) & Table V, line 13 & 1 \newline 1 \\ \hline

5 & $[127,11,79]^2_2$ \newline  $[128,11,80]^2_2$  (Ex) & \text{private} \newline \text{communication \cite{guanPC}}
& 1 \newline 1\\ \hline

6 & $[195,7,139]^2_2$ (Au) \newline $[196,7,140]^2_2$ (ExAu) & Table V, line 16 & 1 \newline 1 \\ \hline

\end{tabular}
\label{table2_verification}
\end{table}

We now examine an additive complementary dual (ACD) code from \cite[Table IV]{guan2023}. 
We note that the alternating form used for quaternary additive codes coincides with the trace Hermitian inner product, see for example \cite{ezermanconcise}. 
For an $\mathbb{F}_4$-linear code, orthogonality with respect to the trace-Hermitian form is equivalent to orthogonality with respect to the usual Hermitian form, because scalar multiplication by elements of $\mathbb{F}_4$ allows one to recover the vanishing of the Hermitian product from the vanishing of all its traces (cf.\ \cite[Theorem 3]{calderbank1998}).
Consequently, if a quaternary ACD code is linear over $\mathbb{F}_4$, its additive dual is identical to its Hermitian dual, and so any such code is a Hermitian linear complementary dual (LCD) code.

Moreover, by Remark~\ref{remarkSLhull}, $\mathrm{SL}$-equivalent codes have the same
hull dimension. In particular, any code $\mathrm{SL}$-equivalent to an ACD code is
again an ACD code. Consequently, if a quaternary ACD code is $\mathrm{SL}$-equivalent
to a linear code $C'$, then $C'$ is a linear ACD code, and hence a Hermitian LCD code.

We are mainly interested in   additive codes whose minimum distance exceeds that of the best known Hermitian LCD codes in \cite[Table 6]{araya2024}. 
 
\begin{example} \label{exampleACD}
 We consider the ACD code in Table IV, line 1 from \cite{guan2023}. The generators are defined in Table VI, line 1 from \cite{guan2023} as follows: 
\[
\begin{aligned}
g(x) &= x^{2} + 1, \\
f_{0}(x) &= x^{19} + x^{14} + x^{11} + x^{10} + x^{2} + 1, \\
f_{1}(x) &= x^{21} + x^{19} + x^{18} + x^{17} + x^{16} + x^{14} + x^{10} + x^{9} + x^{5} + 1.
\end{aligned}
\]
Let $C_{QC}$ be the 1-generator quasi-cyclic code over $\mathbb{F}_2$ generated by $([g(x)f_0(x)], [g(x)f_1(x)])$. Let $C$ be the corresponding ACD code with parameters $[22, 10, 9]^2_2$. The code $C$ is not closed under multiplication by $\omega$ so it is not a linear code over $\mathbb{F}_{4}$. The corresponding matrix $S$ has null space of dimension 2, so Theorem \ref{notSLetest} is inconclusive. 
We therefore apply Step 3 of Algorithm~\ref{equivtest}. 
Let the generator matrix $G$ of $C$ be in reduced row echelon form, so that $G = (\,I_{20} \mid X\,)$. Since $\mathrm{nullity}(S) = 2$, the exhaustive search runs over the $q^d = 4$ elements of $\mathcal{R}$ and yields
\[
J = \operatorname{diag}(\mathbf{M}_\omega, \mathbf{M}_{\omega^2}, 
\mathbf{M}_\omega, \mathbf{M}_{\omega^2}, \dots, 
\mathbf{M}_\omega, \mathbf{M}_{\omega^2}) \in \mathbb{F}_2^{20 \times 20},
\]
where $\mathbf{M}_{\omega^2} = \mathbf{M}_\omega^2$, satisfying $f(J) = \mathbf{0}$. By Theorem~\ref{minimalpolyR}, $C$ is equivalent to a linear code over $\mathbb{F}_4$. 
Explicitly, let $A = \mathrm{diag}(A_1, \dots, A_{22})$, where 
$$
A_i = \begin{cases}
\begin{pmatrix} 1 & 0 \\ 0 & 1 \end{pmatrix} & \text{for odd } i, \\[10pt]
\begin{pmatrix} 0 & 1 \\ 1 & 0 \end{pmatrix} & \text{for even } i.
\end{cases}
$$ 
Then $J(GA) = GA\,M_\omega$, and $\phi^{-1}(GA)$ generates an $\mathbb{F}_4$-linear code  $C'$
with generator matrix

\[ 
\left[
\begin{array}{cccccccccccccccccccccc}
1&0&0&0&0&0&0&0&0&0&\omega^2&1&\omega^2&1&1&1&\omega^2&0&\omega&0&1&1\\
0&1&0&0&0&0&0&0&0&0&\omega^2&\omega^2&\omega&\omega^2&0&0&\omega&\omega&\omega&\omega^2&1&0\\
0&0&1&0&0&0&0&0&0&0&0&\omega&\omega&\omega^2&\omega&0&0&\omega^2&\omega^2&\omega^2&\omega&1\\
0&0&0&1&0&0&0&0&0&0&\omega^2&1&0&\omega&\omega^2&\omega&\omega^2&0&0&\omega&\omega^2&\omega\\
0&0&0&0&1&0&0&0&0&0&\omega&1&\omega^2&\omega^2&0&1&1&\omega&1&0&0&1\\
0&0&0&0&0&1&0&0&0&0&\omega^2&\omega&\omega&\omega^2&\omega^2&1&\omega&1&1&1&1&1\\
0&0&0&0&0&0&1&0&0&0&\omega^2&\omega^2&0&\omega&\omega^2&\omega^2&\omega&\omega^2&\omega^2&1&0&0\\
0&0&0&0&0&0&0&1&0&0&0&\omega&\omega&0&\omega^2&\omega&\omega&\omega^2&\omega&\omega&1&0\\
0&0&0&0&0&0&0&0&1&0&0&0&\omega^2&\omega^2&0&\omega&\omega^2&\omega^2&\omega&\omega^2&\omega^2&1\\
0&0&0&0&0&0&0&0&0&1&\omega^2&1&\omega^2&\omega^2&\omega^2&1&0&\omega&0&\omega^2&\omega^2&\omega^2
\end{array}
\right] 
\]

It can be checked that $C'$ is a Hermitian LCD code with parameters $[22,10,9]_4$.
The best previously known minimum weight for a quaternary Hermitian LCD $[22,10]$ code was $8$ \cite[Table 6]{araya2024}. 
The code $C'$ improves this to $9$. 
\end{example}

 In addition, we verified  all quaternary ACD codes listed in \cite[Table IV]{guan2023} to identify those that potentially improve upon the best known Hermitian LCD codes provided in \cite[Table 6]{araya2024}. We note that the majority of the ACD codes presented in \cite{guan2023} have  minimum distances that are either equal to or strictly lower than the established bounds for linear Hermitian LCD codes; consequently, we omit these codes.
 
We focus on  four additive codes where the reported parameters exceed the known linear bounds. These candidates, corresponding to lines 1, 5, 6, and 10 of \cite[Table IV]{guan2023}, are summarized in Table \ref{tableacd_verification}. As shown in Example \ref{exampleACD}, the code in line 1 is $\textrm{SL}$-equivalent to a linear code. The code in line 6 is derived from the base code in line 5 via the ACD shortening technique described in \cite[Lemma 10]{guan2023}.

Our computations confirm that the codes in lines 5, 6, and 10 have $\text{nullity}(S)=1$. By Theorem \ref{notSLetest},  they are not equivalent to any linear code over $\mathbb{F}_4$.

\begin{table}[H]
\centering
\caption{Verification of Quaternary ACD Codes from \cite[Table IV]{guan2023}}
\renewcommand{\arraystretch}{1.4}
\begin{tabular}{|c|c|p{4cm}|c|}
\hline
\textbf{Line} & \textbf{Code Parameters} & \textbf{Description in \cite{guan2023}} & $\mathrm{nullity}(S)$ \\ \hline
1 & $[22, 10, 9]^2_2$ & Table VI, line 1 & 2 \\ \hline
5 & $[27,10,12]^2_2$  & Table VI, line 4 & 1 \\ \hline
6 & $[26,9,12]^2_2$  & Shortened code of $[27,10,12]^2_2$  & 1 \\ \hline
10 & $[28,12,11]^2_2$   & Table VI, line 7 & 1 \\ \hline
\end{tabular}
\label{tableacd_verification}
\end{table}

We also apply our equivalence test to the quaternary ACD codes constructed in \cite{zhu2024}. 
These codes, obtained via the Plotkin sum construction, were shown in \cite{zhu2024} to outperform the best-known quaternary Hermitian LCD codes. 
Our computations confirm that all ten codes listed in \cite[Table 2]{zhu2024} are strictly additive. 
Five of these were identified by the rank condition in Theorem~\ref{eigenvalue}.
For the remaining five codes, we have $\mathrm{nullity}(S) = 2$, and  as in 
Example~\ref{exampleACD}, we apply Step 3 of Algorithm~\ref{equivtest}. 
The  exhaustive search over the $q^d = 4$ elements of $\mathcal{R}$ shows that no 
element satisfies $f(J) = \mathbf{0}$, so by Theorem~\ref{minimalpolyR} these 
codes are strictly additive.
The verification is summarized in Table \ref{tableACDzhu}. 

\begin{table}[H]
\centering
\caption{Verification of Quaternary ACD Codes from \cite[Table 2]{zhu2024}}
\renewcommand{\arraystretch}{1.4}
\begin{tabular}{|c|c|c|c|}
\hline
\textbf{Line} & \textbf{Code Parameters} & $\mathrm{nullity}(S)$ & \textbf{Method} \\ \hline
1 & $[46, 23, 11]^2_2$ & 2 & Step 3 of Algorithm~1 \\ \hline
2 & $[52, 26, 10]^2_2$ & - & $\operatorname{rank}(G_{10})=1$ \\ \hline
3 & $[56, 28, 12]^2_2$ & 2 & Step 3 of Algorithm~1 \\ \hline
4 & $[58, 29, 12]^2_2$ & - & $\operatorname{rank}(G_{30})=1$ \\ \hline
5 & $[62, 31, 12]^2_2$ & - & $\operatorname{rank}(G_{14})=1$ \\ \hline
6 & $[63, 56, 4]^2_2$ & 2 & Step 3 of Algorithm~1 \\ \hline
7 & $[64, 32, 12]^2_2$ & - & $\operatorname{rank}(G_{18})=1$ \\ \hline
8 & $[70, 35, 14]^2_2$ & 2 & Step 3 of Algorithm~1 \\ \hline
9 & $[72, 36, 15]^2_2$ & 2 & Step 3 of Algorithm~1 \\ \hline
10 & $[74, 37, 14]^2_2$ & - & $\operatorname{rank}(G_{20})=1$ \\ \hline
\end{tabular}
\label{tableACDzhu}
\end{table}

\section{Conclusion}

In this paper, we developed an algebraic framework for determining whether an additive code over $\mathbb{F}_{q^2}$ is equivalent to a linear code. Our main contributions are the following.

\begin{enumerate}

\item We characterized when an additive code is equivalent to a linear code over $\mathbb{F}_{q^2}$ (Theorem~\ref{additivecharacterization}), and reformulated this characterization in terms of the algebra $\mathcal{R}$: under the rank-2 assumption, $C$ is equivalent to a linear code if and only if $\mathcal{R}$ contains an element $J$ with $f(J) = \mathbf{0}$ (Theorem~\ref{minimalpolyR}).

\item We derived two  sufficient criteria for strict additivity: the rank condition of Theorem~\ref{eigenvalue}, and the odd-nullity condition on the matrix $S$ (Theorem~\ref{notSLetest}). Together with the search in Theorem~\ref{minimalpolyR}, these yield a deterministic test for equivalence to a linear code. The test is presented as Algorithm~\ref{equivtest}.

\item We applied our test to verify that all quaternary additive codes from \cite[Table~II]{guan2023} are strictly additive (Table~\ref{table2_verification}), and that all ten quaternary ACD codes from \cite[Table~2]{zhu2024} are strictly additive (Table~\ref{tableACDzhu}).

\item We demonstrated that the $[22, 10, 9]^2_2$ ACD code from \cite[Table~IV]{guan2023} is equivalent to a $[22, 10, 9]_4$ Hermitian LCD code, improving the corresponding entry in \cite[Table~6]{araya2024}.
\end{enumerate}

\section*{Acknowledgements}
We would like to thank C.\ Guan, J.\ Lv, G.\ Luo, and Z.\ Ma for providing the explicit generator polynomials for the $[127, 11, 79]^2_2$ code and its extended code. 
We also wish to thank the reviewers for their careful reading and constructive comments, which have greatly improved the presentation of the paper.
K. Abdukhalikov was supported by UAEU-AUA grant G00005533.

\end{document}